\documentclass[11pt,a4paper,oneside]{article}
\usepackage[T1]{fontenc}
\usepackage[ansinew]{inputenc}
\usepackage[english]{babel}
\usepackage{amsfonts}
\usepackage{amsmath}
\usepackage{bm}
\usepackage{array}
\usepackage{amsthm}
\usepackage{amssymb}
\usepackage{graphicx}
\usepackage{braket}
\usepackage{verbatim}
\usepackage[table]{xcolor}
\usepackage{caption}
\usepackage{cite}
\usepackage{textcomp}
\usepackage{url}
\usepackage{hyperref}
\hypersetup{
    colorlinks=true,
    linkcolor=black,
    citecolor=black
    }
\usepackage{multicol}
\newcommand{\be}{\begin{equation}}
\newcommand{\ee}{\end{equation}}
\definecolor{pinegreen}{rgb}{0.0, 0.47, 0.44}

\renewcommand{\d}{\mbox{${\rm d}$}}
\newtheorem{theorem}{Theorem}

\theoremstyle{definition}

\theoremstyle{remark}
\newtheorem{remark}{Remark}

\newcommand{\eFE}{e^{(E)}_F}
\newcommand{\eVs}{e^{(V)}_\sigma}
\newcommand{\eV}{e^{(V)}}
\newcommand{\rstar}{\rho^*}

\title{\bf Energy of a  non-linear viscoelastic model\\ 
compatible with fractional relaxation} %
\author{
Andrea~Giusti$^{1,2}$
, Andrea~Mentrelli$^{2,3,4}$, and Tommaso~Ruggeri$^{2,3,5}$ 
\\
\\
$^1$ {\em Institute for Theoretical Physics, ETH Zurich}
\\
{\em Wolfgang-Pauli-Strasse 27, 8093 Zurich, Switzerland}
\\
\\
$^2${\em Alma Mater Research Center on Applied Mathematics (AM$^2$)}
		\\
{\em Via Saragozza 8, 40123 Bologna, Italy}
\\
\\
$^3${\em Department of Mathematics, University of Bologna}
	\\
	{\em Piazza di Porta San Donato 5, 40126 Bologna, Italy} 
\\
\\
$^4${\em Istituto Nazionale di Fisica Nucleare (I.N.F.N.), Sezione di Bologna, I.S. FLAG}
	\\
	{\em Viale Berti Pichat 6/2, 40127 Bologna, Italy}
\\
\\
$^5${\em Accademia Nazionale dei Lincei, Rome, Italy}
}
\begin{document}
\maketitle
\begin{abstract}
Recently, a non-linear model of viscoelasticity based on Rational Extended Thermodynamics was proposed in [T. Ruggeri, Int. J. Non-Linear Mech. 160, 104658 (2024)]. This theory extends the evolution of the viscous stress beyond the linear framework of the Maxwell model to the non-linear realm, provided that the viscous energy function is given.
This work aims at establishing a possible constitutive law for the viscous energy such that the relaxation modulus of the fractional Maxwell model of order $\alpha \in (1/2, 1]$ is contained within the solutions of the (non-linear) relaxation experiment. Necessary and sufficient conditions for the existence of this coincident solution are discussed, together with a numerical evaluation of the viscous energy associated with the nonlinear model.
\end{abstract}

\newpage

\section{Introduction} 
Viscoelasticity is an intriguing research topic at the crossroad between applied mathematics 
and engineering, with many practical applications in material science \cite{Christensen, Amabili}.
The {\em Maxwell model} is the prototypical example of a linear viscoelastic model with exponential relaxation \cite{Christensen}. However, with the aim of describing anomalous materials, while preserving the linearity of the constitutive laws, many generalisation of standard models of linear viscoelasticity involving fractional derivatives have been proposed over the past few decades (see, {\em e.g.}, \cite{MainardiBook} and references therein). These non-local modifications lead to memory functions displaying a power-law decay, rather than an exponential one, thus entailing the emergence of long-memory effects \cite{MainardiBook}.  

In \cite{ViscoRuggeri}, T.~Ruggeri introduced a local non-linear viscoelastic model within the framework of Rational Extended Thermodynamics (RET) \cite{RET, beyond, newbook}. This model is uni-axial (thus effectively one-dimensional in space) and assumes that the considered process is isothermal. The resulting differential system, obtained by means of the universal principles of RET, reads \cite{ViscoRuggeri}:
\begin{align}
\label{nnelastvisco}
\begin{split}
 & \rho^* \frac{\partial v}{\partial t} - \frac{\partial }{\partial X}(T(F) + \sigma) = \rho^* b,   \\
 & \frac{\partial F}{\partial t} - \frac{\partial v}{\partial X} = 0, \\
 & \frac{\partial }{\partial t} (Z(\sigma)-F)  = -\frac{\sigma}{\mu(F)},  
  \end{split}
\end{align}
where
\begin{equation*}
    T(F) = \rstar \eFE(F), \quad  \quad Z(\sigma)= \int \rstar \frac{\eVs(\sigma)}{\sigma} d\sigma \, ,
\end{equation*}
adopting the notation according to which a subscript variable denotes differentiation with respect to the corresponding variable. Note that $\rho^*$ denotes the  mass density in the  reference frame, $v$ is the velocity, $T$ denotes the first Piola-Kirchhoff elastic stress tensor, $\sigma$ represents the viscous stress, $F$ denotes the  deformation gradient, $b$ is the external body force, and, lastly,  $\mu$ represents the viscous coefficient. Furthermore, $e(F,\sigma)$ is the internal energy, which splits into the sum
of an elastic part $e^{(E)}(F)$ and a viscous part $e^{(V)}(\sigma)$.

Any solution of \eqref{nnelastvisco} satisfies also the supplementary energy balance equation
 \begin{equation*}
     \rho^* \frac{\partial }{\partial t}\left(\frac{v^2}{2} + e(F,\sigma)\right) - \frac{\partial }{\partial X}\Big((T(F) + \sigma)v\Big) - \rho^* b \, v = \mathcal{E} =-\frac{\sigma^2}{\mu(F)} \leq 0 \, ,
 \end{equation*}
that in the isothermal case corresponds to the \emph{entropy principle}.

The system in Eq.~\eqref{nnelastvisco} is  symmetric hyperbolic   provided that the following inequalities are satisfied \cite{ViscoRuggeri}:
\begin{eqnarray}\label{dis}
	e^{(E)}_{FF}(F) >0 \, , \qquad \frac{\eVs(\sigma)}{\sigma} >0 \, , \qquad \mu(F) \geq 0 \, .
\end{eqnarray}
In one spatial dimension all quantities are scalars and the deformation gradient reduces to 
\begin{equation*}
	F=\sqrt{1+2\varepsilon} \, ,
\end{equation*}
where $\varepsilon$ is the deformation,
as detailed in \cite{ViscoRuggeri}. The last line in Eq.~\eqref{nnelastvisco} can therefore be rewritten, for classical solutions, as
\begin{equation}\label{Eq:Ruggeri}
\frac{\sigma}{\mu (\varepsilon )} + \rho^\ast \frac{\eVs(\sigma)}{\sigma} \dot{\sigma} = \frac{\dot{\varepsilon }}{\sqrt{1+ 2 \varepsilon }} \, ,
\end{equation}
where the dot denotes the derivative with respect to time. Notably, Eq.~\eqref{Eq:Ruggeri} looks like a non-linear modification of the standard Maxwell model, equipped with a  ``{\em non-linear relaxation time}''
\begin{equation}\label{tautau}
    \tau(\varepsilon ,\sigma) := \rho^\ast \mu(\varepsilon ) \, \frac{\eVs(\sigma)}{\sigma} 
\end{equation}
which is entirely determined, except for the viscous coefficient $\mu(\varepsilon )$, upon identification of the viscous energy $\eV(\sigma)$. 

As pointed out in \cite{ViscoRuggeri}, it is now crucial to compare the properties of the model given in \eqref{nnelastvisco}, with particular regard for Eq.~\eqref{Eq:Ruggeri}, with other well-established models in the literature as well as against experimental data. This procedure will allow one to classify the set of viscoelastic models compatible with the proposed local non-linear theory in \eqref{Eq:Ruggeri} and, consequently, determine the associated viscous energy within this general scheme.

In this work, we investigate for which viscous energy constitutive equation $\eV(\sigma)$ the \emph{relaxation modulus}, {\em i.e.}, the solution of the stress relaxation experiment, for a given constant strain $\varepsilon (t) = \varepsilon_0$ of \eqref{Eq:Ruggeri}, of the proposed non-linear model is also a solution of the linear fractional Maxwell model, as defined by the relation \cite{MainardiSpada}:
\be
\label{eq:fracmax}
\sigma + \tau_0^\alpha \, D^\alpha \sigma = b\, D^\alpha \varepsilon  \, , \qquad \alpha \in (0,1) \, ,
\ee
where $\tau_0$ and $b$ are real dimensional constants and $D^\alpha$ denotes the Caputo fractional derivative with respect to time.

\section{Relaxation modulus and viscous energy}
The relaxation modulus for the local non-linear model in Eq.~\eqref{Eq:Ruggeri} is the solution of
\be \label{nonlineare}
 \sigma  + \rho^\ast {\mu (\varepsilon _0)}\frac{e^{(V)}_\sigma (\sigma)}{\sigma}\, \dot{\sigma} = 0 \, .
\ee
For the linear fractional Maxwell model of order $\alpha \in (0,1)$, the relaxation modulus is a solution of
\be \label{frazione}
\sigma + \tau_0^\alpha \, D^\alpha \sigma = 0  \, , \qquad \alpha \in (0,1) \, ,
\ee 
with initial condition $\sigma(0) = \sigma_0$,  which is given by
\begin{equation}
    \label{eq:fracrelax}
   \sigma (t) = \sigma_0 \, E_\alpha \big[ - (t/\tau_0)^\alpha \big] \, ,
\end{equation}
where 
$$
E_\beta (z) := \sum_{k=0}^\infty \frac{z^k}{\Gamma (\beta \, k + 1)} \, , \quad \beta > 0 \, , \quad z \in \mathbb{C} \, ,
$$ 
denotes the {\em Mittag-Leffler function} \cite{MainardiML}, and $\Gamma (z)$ represents Euler's Gamma function. 

Recalling that the function $E_\alpha (-t^\alpha)$ is {\em positive} on $\mathbb{R}^+$, {\em completely monotonic} (see \cite{MainardiML}), and such that
$$
\frac{\d}{\d t} \big( E_\alpha (-t^\alpha) \big) = - t^{\alpha - 1} \, E_{\alpha, \alpha} (-t^\alpha) \, ,
$$
where
$$
E_{\beta, \delta} (z) := \sum_{k=0}^\infty \frac{z^k}{\Gamma (\beta \, k + \delta)} \, , \quad {\rm Re} (\beta) > 0 \, , \quad \delta \in \mathbb{C} \, , \quad z \in \mathbb{C} \, ,
$$ 
denotes the {\em two-parameter Mittag-Leffler function}\footnote{For a comprehensive literature review on Mittag-Leffler functions and their role in fractional caluclus we refer the reader to the monograph by R. Gorenflo, A. A. Kilbas, {\em et al.} \cite{MLsBook}.}, we can prove the following: 

\smallskip 

\begin{theorem} \label{thm-1}
    Consider the family of constitutive equations $\eV(\sigma)$ depending on $\alpha$ and expressed in the following parametric form:
    \begin{align}\label{teor1-1}
    \begin{split}
    e^{(V)}(s) &= 
    e_0 - \frac{k_0^2}{\rho^*\mu(\varepsilon_0)}\int_0^s \Big(E_\alpha \big[-(\bar{s}/\tau_0)^\alpha \big]\Big)^2 \, d \bar{s} \, ,  \\
    \sigma(s) & = k_0 \, E_\alpha \big[-({s}/\tau_0)^\alpha \big] \, , 
    \end{split}
    \end{align}  
    with $s\geq 0$, $e_0$ being a real inessential constant, and $k_0$ a structural constant of the material. Then, there exists a solution $\sigma(t) = k_0  \, E_\alpha \big[ - (t/\tau_0)^\alpha \big]$ of both Eqs.~\eqref{nonlineare} and \eqref{frazione} if we choose as initial condition $\sigma(0) = k_0 $. 
    The viscous energy $\eV(\sigma)$ is bounded $\forall \sigma \in [0,k_0]$ if and only if $E_\alpha[-(s/\tau_0)^\alpha] \in L^2\big( [0,\infty) \big)$, namely, if and only if $\alpha \in (1/2, 1)$. 
    The non-linear relaxation time defined in Eq.~\eqref{tautau}, computed assuming $\varepsilon  (t) = \varepsilon _0$, is positive for all $\sigma \in (0, k_0 )$, diverges at $\sigma = 0$, and vanishes at $\sigma = k_0 $.
    The constant $e_0$ can be chosen to be equal to
    \begin{equation}\label{e00}
    e_0 = \frac{k_0 ^2}{\rho^*\mu(\varepsilon_0)}\int_0^\infty \Big(E_\alpha \big[-(\bar{s}/\tau_0)^\alpha \big]\Big)^2 \, d \bar{s}, 
    \end{equation}
    so that
    $$\lim_{\sigma \rightarrow  0} e^{(V)}(\sigma) =0 \, . $$
\end{theorem}

\smallskip

\begin{proof}
First, for \eqref{teor1-1} we observe that for $\eV(s)$ to be finite for all $s \in [0, +\infty)$ one needs $E_\alpha[-(s/\tau_0)^\alpha] \in L^2\big( [0,\infty) \big)$. Recalling the asymptotic behaviour of the the Mittag-Leffler function \cite{MainardiML}, {\em i.e.},
\be \label{asym}
E_\alpha(-x^\alpha)  \sim 
\left\{ 
\begin{aligned}
& \exp \left[- x^\alpha/\Gamma (\alpha + 1) \right] \, , \,\, \mbox{as} \,\, x \to 0^+ \, ,\\
& x^{-\alpha}/\Gamma (1-\alpha) \, , \,\,\, \mbox{as}\,\, x \to + \infty \, , 
\end{aligned}
\right. \sim 
\left\{ 
\begin{aligned}
& 1 - x^\alpha/\Gamma (\alpha + 1) \, , \,\, \mbox{as} \,\, x \to 0^+ \, ,\\
& x^{-\alpha}/\Gamma (1-\alpha) \, , \,\,\, \mbox{as}\,\, x \to + \infty \, ,
\end{aligned}
\right.
\ee
one can easily conclude that $E_\alpha[-(s/\tau_0)^\alpha] \in L^2\big( [0,\infty) \big)$ if and only if $\alpha \in (1/2, 1)$.

Then, fixing $\alpha \in (1/2, 1)$, from Eq.~\eqref{teor1-1} we have that
$$
{e}^{(V)} _s (s) = - \frac{\sigma^2(s)}{\rho^\ast \mu (\varepsilon _0)} \, ,
$$
and hence, by means of the chain rule, we find
$$
{e}^{(V)} _\sigma  (\sigma(s)) = - \frac{\sigma^2(s)}{\rho^\ast \mu (\varepsilon _0) \, \sigma_s (s)} \, ,
$$
which implies
\be
\label{eq:intermedia-1}
\frac{e^{(V)} _\sigma (\sigma (s))}{\sigma (s)} = - \frac{\sigma (s)}{\rho^\ast \mu (\varepsilon _0) \, \sigma_s (s)} \, .
\ee
Inserting \eqref{eq:intermedia-1} into Eq.~\eqref{nonlineare} we get
\begin{equation*} \label{intermediate}
  \frac{\sigma_s (s)}{\sigma (s)} = \frac{\dot{\sigma} (t)}{\sigma (t)} \, , 
\end{equation*}
thus implying that $s = t + c$, with $c \in \mathbb{R}$. Therefore from  \eqref{teor1-1}$_2$, the solution of \eqref{nonlineare} with inital data $\sigma(0)=\sigma_0$ is
\be \label{sR}
\sigma(t) = k_0 \, E_\alpha \left[ -\left(\frac{t+c}{\tau_0}\right)^\alpha \right] \, \quad \text{with} \,\, c \,\,\text{solution of} \quad E_\alpha \left[ -\left(\frac{c}{\tau_0}\right)^\alpha \right]= \frac{\sigma_0}{k_0} .
\ee
If we then compare the latter with \eqref{eq:fracrelax} we conclude that the two solutions coincide if and only if $c=0$ and therefore  only if we choose  the initial condition $\sigma_0 = k_0$.

Second, since $E_\alpha (-t^\alpha)$ is positive and strictly decreasing on $\mathbb{R}^+$ we can conclude that
$$
\frac{\tau(\varepsilon_0 ,\sigma(s))}{\rho^\ast \mu(\varepsilon_0)} = \frac{\eVs(\sigma(s))}{\sigma(s)} = \frac{\tau_0}{\rho^\ast \, \mu (\varepsilon _0)}
\frac{E_\alpha \big[ - (s/\tau_0)^\alpha \big]}{(s/\tau_0)^{\alpha - 1} \, E_{\alpha, \alpha} \big[-(s/\tau_0)^\alpha\big]}
$$
is positive and finite for any $s \in (0, +\infty)$,  while it vanishes at $s=0$.
Equivalently from \eqref{teor1-1}$_2$, we have shown that $\tau(\varepsilon_0 ,\sigma)$ is positive and finite for any $\sigma \in (0, k_0 )$ and that it vanishes at $\sigma = k_0 $. 

Lastly, if one recalls that \cite{MainardiML} 
$$
x^{\alpha - 1} \, E_{\alpha, \alpha} \big(-x^\alpha\big) \sim \frac{\sin (\alpha \pi) \, \Gamma (\alpha +1)}{\pi} \, x^{-\alpha -1} \, , \quad \mbox{as} \,\, x \to +\infty \, ,
$$
then
$$
\frac{\tau(\varepsilon_0 ,\sigma(s))}{\rho^\ast \mu(\varepsilon_0)} = \frac{\eVs(\sigma(s))}{\sigma(s)} \sim  \frac{1}{\alpha} \frac{s}{\rho^\ast \mu(\varepsilon_0)}  \quad \mbox{as} \,\, s \to +\infty \, ,
$$
or, equivalently,
$$
\lim_{\sigma \to 0}
\frac{\tau(\varepsilon_0 ,\sigma )}{\rho^\ast \mu(\varepsilon_0)} = 
\lim_{\sigma \to 0} \frac{\eVs(\sigma)}{\sigma} = +\infty 
$$
and the proof is completed.
\end{proof}

\begin{remark}
It is important to observe that the convexity condition \eqref{dis}$_2$, corresponding to $\tau(\varepsilon_0 , \sigma) > 0$, is valid on a finite interval $\sigma \in (0, k_0) $. This is not surprising since in many physical scenarios convexity is known to hold only for a subset of the values of the fields.
\end{remark} 
\begin{remark}
Let us consider the case $\alpha=1$, that corresponds to the ordinary Maxwell model. By definition $E_1 (z) = {\rm e}^z$, hence from \eqref{teor1-1} we find that
$$
e^{(V)}(s) = \frac{\tau_0}{2 \rho^\ast \mu (\varepsilon_0)} \, k_0 ^2 \, {\rm e}^{-\frac{2 s}{\tau_0}} \, , \qquad \sigma(s) = k_0 \, {\rm e}^{-\frac{s}{\tau_0}} \, ,
$$
from which it follows that
$$
e^{(V)}(\sigma) = \frac{\tau_0}{2 \rho^\ast \mu (\varepsilon_0)} \, \sigma^2 \, .
$$
Then we can easily compute the associated non-linear relaxation time that yields 
$$
\tau(\varepsilon_0 ,\sigma) = \rho^\ast \mu(\varepsilon_0) \, \frac{\eVs(\sigma)}{\sigma} = \tau_0 > 0 \, ,
$$
that coincides with the standard relaxation constant time of the linear model, as expected, and both equations \eqref{Eq:Ruggeri} and \eqref{eq:fracmax} reduced with $b=\mu$ to the linear Maxwell model.  
\end{remark} 

\smallskip

In Fig.~\ref{Figura1} we plot the non-dimensional non-linear relaxation time
and the non-dimensional viscous energy as functions of $\sigma / k_0$ choosing the   constant $e_0$ as in Eq.~\eqref{e00}.

\begin{figure}
    \centering
    \includegraphics[width=0.49\linewidth]{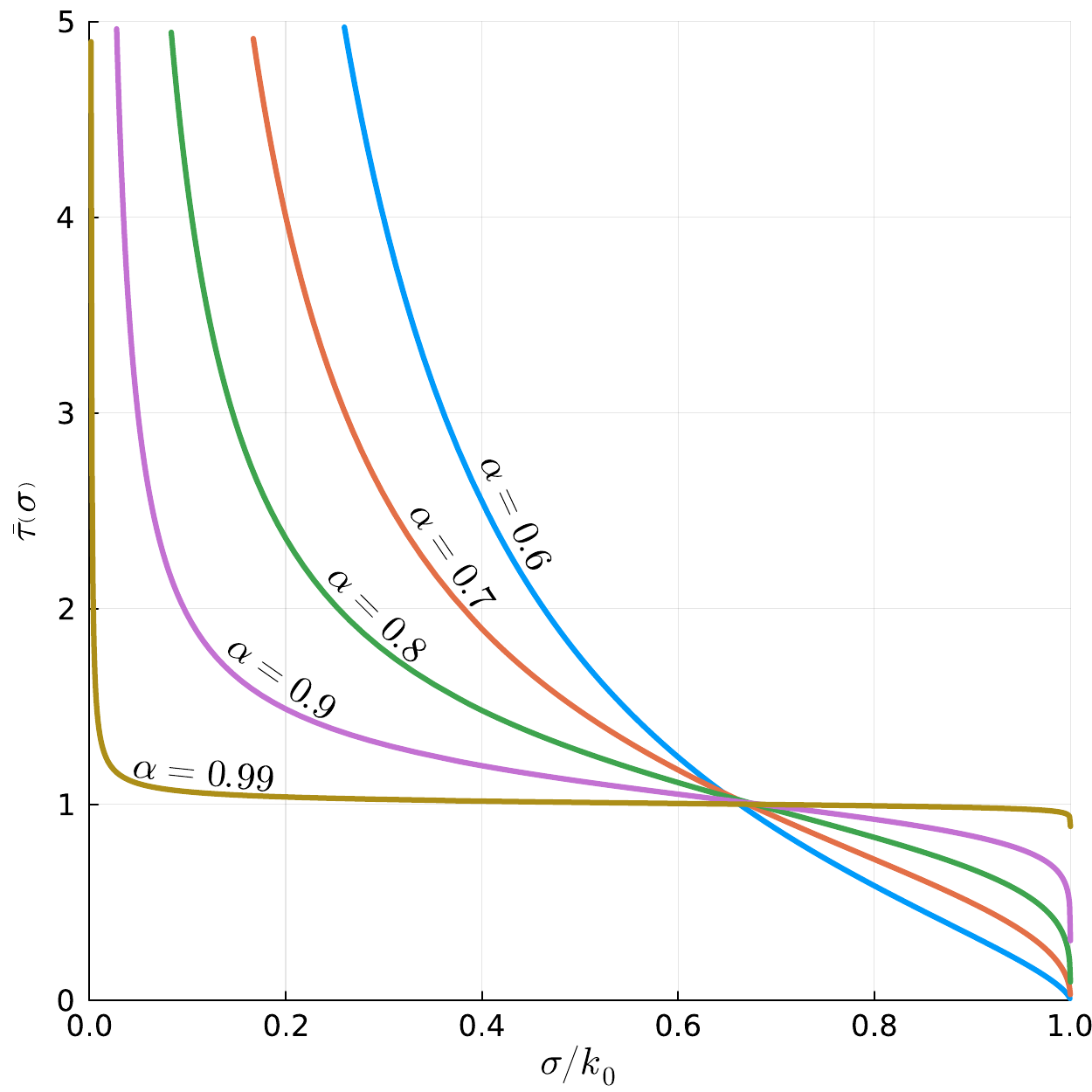}%
    \includegraphics[width=0.49\linewidth]{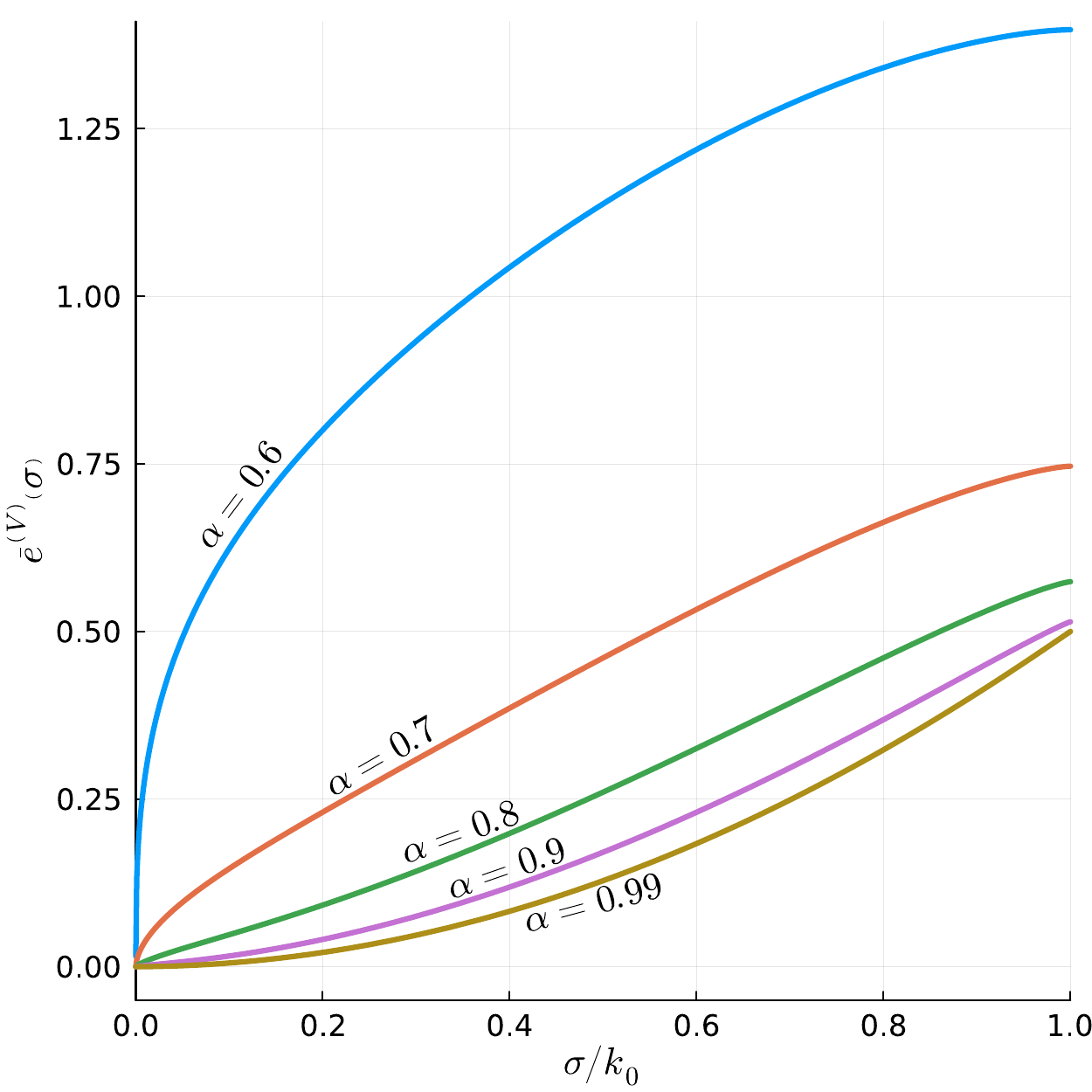}
    \caption{Non-dimensional non-linear relaxation time $\bar{\tau}(\sigma) = \tau/\tau_0$ (left panel) and non-dimensional viscous energy $\bar{e}^{(V)}\left(\sigma\right )= \left[\rho^\ast \mu \left(\varepsilon_0\right)/\left(\tau_0 k_0^2\right)\right] \, e^{(V)}(\sigma)$
    (right panel).}
    \label{Figura1}
\end{figure}
\smallskip

Note that {\bf Theorem \ref{thm-1}} addresses the existence of a viscous energy \eqref{teor1-1} that leads to a coincident solution of both RET-improved viscoelasticity and the fractional Maxwell model. Then, it sets necessary and sufficient conditions for the existence of such a solution, among which one finds that the initial condition for the relaxation experiment of fractional Maxwell model \eqref{eq:fracrelax} must coincide with the structural constant of the material, {\em i.e.}, $\sigma(0) = k_0$.

However, this is not the only solution to the relaxation experiment associated with the constitutive equation \eqref{teor1-1}. In fact if $\sigma_0\neq k_0$ we can show that:
\begin{theorem} \label{thm-2}
Let $\alpha\in(1/2, 1)$ and $c>0$. There exist solutions of the non-linear local RET model Eq.~\eqref{nonlineare} with initial data $\sigma(0)=\sigma_0 \in (0, k_0)$ of the form given in \eqref{sR} that we rewrite as
\be
\label{sol-RET}
\sigma_R(t) = k_0 \, E_\alpha \left[ -\left(\frac{t+c}{\tau_0}\right)^\alpha \right] \, \quad \text{with} \,\, c \,\,\text{solution of} \quad E_\alpha \left[ -\left(\frac{c}{\tau_0}\right)^\alpha \right]= \frac{\sigma_0}{k_0} .
\ee
While the solution of the fractional equation \eqref{frazione} is given in \eqref{eq:fracrelax} that we rewrite as:
\begin{equation}\label{sol-F}
   \sigma_F (t) = \sigma_0 \, E_\alpha \big[ - (t/\tau_0)^\alpha \big] \, .
\end{equation}
The two solutions \eqref{sol-RET} and \eqref{sol-F} are different (except at initial time) but the solution of the non-linear local RET  model has lower and upper bounds that depends on the solution of the fractional equation according the following inequalities:
    \be \label{upper}
    \sigma_{\rm F} (t) < \sigma_{\rm R} (t) < \frac{k_0}{\sigma_0} \, \sigma_{\rm F} (t) \, , \quad \forall t > 0 \,  . 
    \ee
\end{theorem}

\begin{proof}
    Because of the monotonicity of the Mittag-Leffler function, the requirement of $c>0$ implies $\sigma_0 < k_0$. Furthermore, for any $t>0$ we have that $\sigma_{\rm R} (t)$ and $\sigma_{\rm F} (t)$ are monotonically decreasing, and taking into account that for large $t$ are valid the \eqref{asym}, we have
    $$
    \lim_{t\rightarrow \infty} \frac{\sigma_R(t) }{\sigma_F(t)}=\lim_{t\rightarrow \infty} \frac{k_0}{\sigma_0}\left(1+\frac{c}{t}\right)^{-\alpha}= \frac{k_0}{\sigma_0} > 1.
    $$
    and hence $\sigma _{\rm F} (t) < \sigma _{\rm R} (t), \, \forall t>0$. Furthermore, again since $E_\alpha \big( - x^\alpha \big)$ is decreasing on $\mathbb{R}^+$ we have that
    $$
    \sigma _{\rm R} (t) = k _0 \, E_\alpha \left[-\left( \frac{t+c}{\tau_0}\right)^\alpha \right] <  k _0 \, E_\alpha \left[-\left( \frac{t}{\tau_0}\right)^\alpha \right] = \frac{k_0}{\sigma_0} \, \sigma _{\rm F} (t) \, ,
    $$
    which concludes the proof.
\end{proof}

\smallskip

In Fig.~\ref{Figura2} we plot the comparison between $\sigma_R(t)$ and $\sigma_F(t)$, as well the upper bound $\sigma^{ub} = k_0 \,\sigma_F/\sigma_0 $ for $\sigma_R(t)$ discussed in {\bf Theorem \ref{thm-2}}.
\begin{figure}
    \centering
    \includegraphics[width=0.49\linewidth]{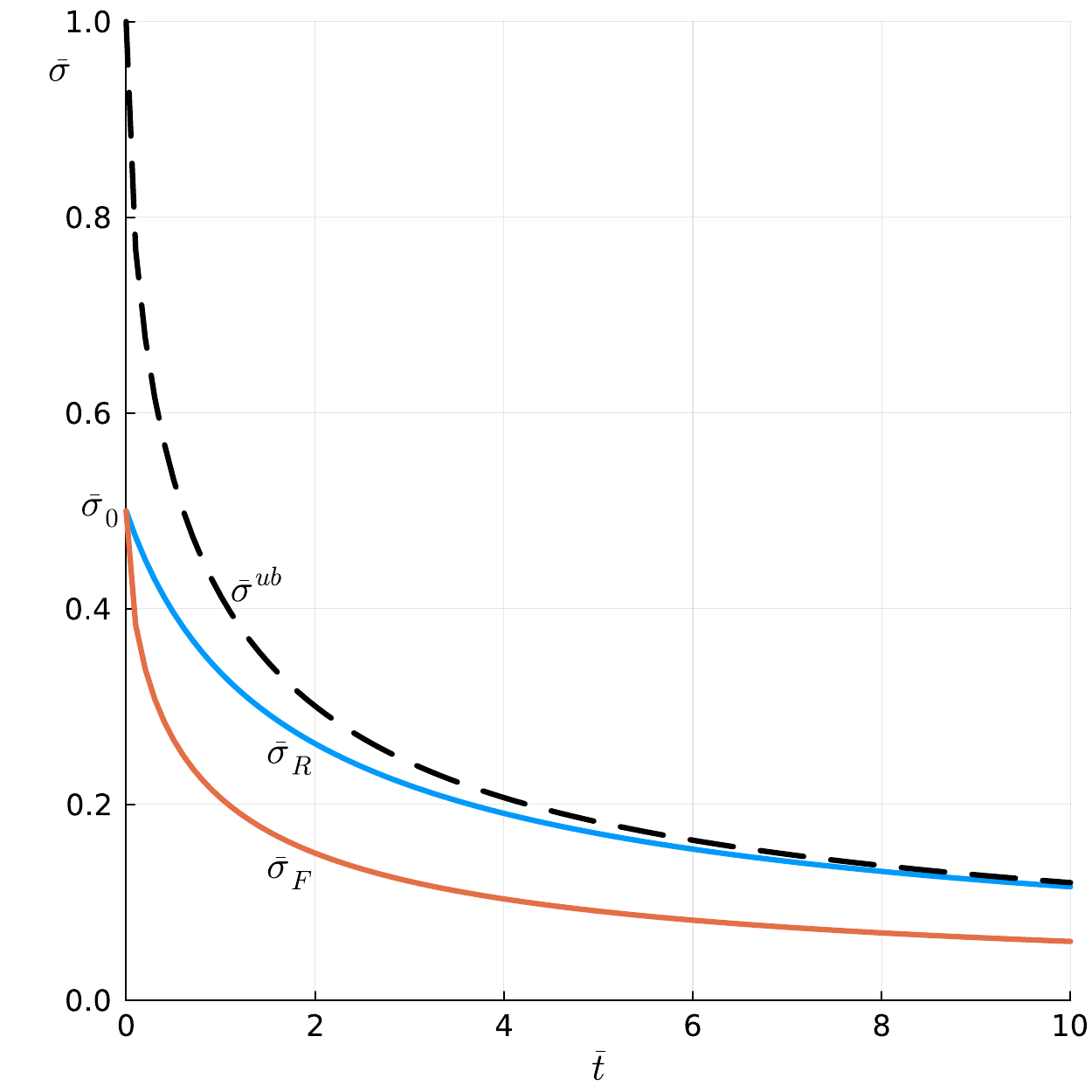}
    \caption{$\bar\sigma = \sigma/k_0$ and the dimensionles upper bound $\bar\sigma^{ub} = \left(k_0/\sigma_0\right) \sigma_F/k_0 = \sigma_F/\sigma_0$ as a function of $\bar t = t / \tau_0$, for $\alpha = 0.6$, $k_0 = 1$, $\sigma_0 = 1/2$.}
    \label{Figura2}
\end{figure}
\section{Discussion}
We have determined the  constitutive law for the viscous energy \eqref{thm-1} for the RET-improved theory of viscoelasticity, proposed in \cite{ViscoRuggeri},
such that the relaxation modulus of the fractional Maxwell model with order $\alpha \in (1/2, 1]$, {\em i.e.} \eqref{eq:fracrelax}, is contained within the solutions of the (non-linear) relaxation experiment \eqref{nonlineare}.

In {\bf Theorem \ref{thm-1}} we have shown that, given the viscous energy \eqref{teor1-1}, the solution of the non-linear relaxation experiment coincides with the one of the fractional Maxwell model of order $\alpha$ if and only if $\alpha \in (1/2, 1]$ and the fractional model provided that we choose $\sigma (0) = k_0$. The viscous energy associated to the non-linear model reproducing such a scenario is finite for all values of $\sigma \in [0, k_0]$, while the non-linear relaxation time $\tau (\varepsilon_0, \sigma)$ turns out to be positive and finite for all $\sigma \in (0,k_0)$, vanishing at $\sigma = k_0$, and diverging at $\sigma = 0$. This implies that the convexity condition is only satisfied for $\sigma \in (0,k_0)$.

The non-local relaxation experiment \eqref{nonlineare} also admits solutions other than the one associated to the fractional Maxwell model, given the same initial condition. Nonetheless, the relaxation modulus of the fractional Maxwell model of order $\alpha \in (1/2, 1)$ determines upper and lower bounds for the relaxation modulus of the RET-inspired model, as discussed in {\bf Theorem \ref{thm-2}}.

\bigskip

\noindent {\bf Acknowledgments}:
This work has been carried out in the framework of the activities of the Italian National Group of Mathematical Physics [Gruppo Nazionale per la Fisica Matematica (GNFM), Istituto Nazionale di Alta Matematica (INdAM)]. 
A.~Mentrelli is partially funded by the European Union -- NextGenerationEU under the National Recovery and Resilience Plan (PNRR) - Mission 4 \textit{Education and research}, Component 2 \textit{From research to business} -- Investment 1.1 Notice PRIN~2022 -- DD N.~104 dated 2/2/2022, entitled ``The Mathematics and Mechanics of Non-linear Wave Propagation in Solids'' (proposal code: 2022P5R22A; CUP: J53D23002350006), and by the Italian National Institute for Nuclear Physics (INFN), grant FLAG. A.~Giusti is grateful to the Department of Mathematics and the Alma Mater Research Center on Applied Mathematics (${{\cal A} {\cal M}}^2$) at the University of Bologna for hospitality during the completion of this work.

\end{document}